\documentclass[a4paper,USenglish,cleveref, autoref,thm-restate]{lipics-v2021}

\pdfoutput=1
\hideLIPIcs  
\nolinenumbers

\usepackage{graphicx}
\usepackage[basic]{complexity}
\usepackage{xspace}
\usepackage{xcolor}
\usepackage{csquotes}

 \usepackage{array}
 \usepackage{booktabs}
\hypersetup{
    colorlinks,
    linkcolor={red!50!black},
    citecolor={blue!50!black},
    urlcolor={blue!80!black}
}

\newcommand{\TSAT}{\textsc{3-Sat}\xspace}
\newcommand{\MPTSAT}{\textsc{Monotone-Planar-3-Sat}\xspace}
\newcommand{\PTSAT}{\textsc{Planar-3-Sat}\xspace}

\newcommand{\segRepr}{segment representation\xspace}

\title{Reconfiguration of unit squares and disks:
$\PSPACE$-hardness in simple settings} 

\author{Mikkel Abrahamsen}{University of Copenhagen, Denmark}{miab@di.ku.dk}{https://orcid.org/0000-0003-2734-4690}{}
\author{Kevin Buchin}{TU Dortmund, Germany}{kevin.buchin@tu-dortmund.de}{https://orcid.org/0000-0002-3022-7877}{}
\author{Maike Buchin}{Ruhr University Bochum, Germany}{maike.buchin@rub.de}{https://orcid.org/0000-0002-3446-4343}{}
\author{Linda Kleist}{Universität Potsdam, Germany}{kleist@cs.uni-potsdam.de}{https://orcid.org/0000-0002-3786-916X}{}
\author{Maarten Löffler}{Utrecht University, The Netherlands}{m.loffler@uu.nl}{https://orcid.org/0009-0001-9403-8856}{}
\author{Lena Schlipf}{Universität Tuebingen, Germany}{lena.schlipf@uni-tuebingen.de}{https://orcid.org/0000-0001-7043-1867}{}
\author{André Schulz}{FernUniversit\"at in Hagen, Germany}{andre.schulz@fernuni-hagen.de}{0000-0002-2134-4852}{}
\author{Jack Stade}{University of Copenhagen, Denmark}{jast@di.ku.dk}{https://orcid.org/0009-0007-9153-6589}{}

\authorrunning{Abrahamsen, Buchin, Buchin, Kleist, Löffler, Schlipf, Schulz, Stade} 

\Copyright{Mikkel Abrahamsen, Kevin Buchin, Maike Buchin, Linda Kleist, Maarten Löffler, Lena Schlipf, André Schulz and Jack Stade} 



\category{} 

\relatedversion{}

\acknowledgements{This work was initiated at the Dagstuhl Seminar 24072: Triangulations in Geometry and Topology. We thank the organizers and the participants for the fruitful atmosphere.}

\begin{document}

\maketitle

\begin{abstract}
We study two well-known reconfiguration problems.
Given a start and a target configuration of geometric objects in a polygon, we wonder whether we can move the objects from the start configuration to the target configuration while avoiding collisions between the objects and staying within the polygon.
Problems of this type have been considered since the early 80s by roboticists and computational geometers.
In this paper, we study some of the simplest possible variants where the objects are labeled or unlabeled unit squares or unit disks.
In unlabeled reconfiguration, the objects are identical, so that any object is allowed to end at any of the targets positions. In the labeled variant, each object has a designated target position. The results for the labeled variants are direct consequences from our insights on the unlabeled versions. 

We show that it is $\PSPACE$-hard to decide whether there exists a reconfiguration of (unlabeled/labeled) unit squares even in a simple polygon.
Previously, it was only known to be $\PSPACE$-hard in a polygon with holes for both the unlabeled [Solovey and Halperin, Int. J. Robotics Res. 2016] and labeled version [Brunner, Chung, Demaine, Hendrickson, Hesterberg, Suhl, Zeff, FUN 2021].
Our proof is based on a result of independent interest, namely that reconfiguration between two satisfying assignments of a formula of \MPTSAT is also $\PSPACE$-complete.
The reduction from reconfiguration of \MPTSAT to reconfiguration of unit squares extends techniques recently developed to show $\NP$-hardness of packing unit squares in a simple polygon [Abrahamsen and Stade, FOCS 2024].
We also show $\PSPACE$-hardness of reconfiguration of (unlabeled/labeled) unit disks in a polygon with holes.
Previously, it was  known that unlabeled reconfiguration of  disks of two different sizes was $\PSPACE$-hard [Brocken, van der Heijden, Kostitsyna, Lo-Wong and Surtel, FUN 2021]. 
\end{abstract}

\newpage
\section{Introduction}

In \emph{geometric reconfiguration} we are given a start configuration of geometric \emph{objects} and a desired target configuration, see \cref{fig:intro} for an example. A particularly well-studied geometric configuration problem is \emph{multi-robot motion planning}.
The goal is to move the robots (or objects) beginning in the start configuration in some way, avoiding collisions between the robots and with static obstacles, so that we obtain the target configuration.
Problems of this type have been considered since the early 80s by roboticists and computational geometers.
Hopcroft, Schwartz and Sharir~\cite{doi:10.1177/027836498400300405} showed already in 1984 that reconfiguration of labeled rectangles moving in a rectangle is $\PSPACE$-hard. The motivation of their work was ``to find moving systems (necessarily involving arbitrarily many degrees of freedom) whose geometry is \emph{as simple as possible}, but for which the motion planning problem is still PSPACE-hard''.

\begin{figure}[htb]
\centering
\includegraphics[page=1]{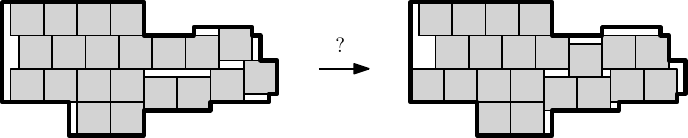}
\caption{Can the left configuration of unit squares in a simple polygon be reconfigured to the right configuration of unit squares?}
\label{fig:intro}
\end{figure}

In this paper, we study some of the simplest possible variants where the objects are unlabeled or labeled axis-aligned congruent squares or unit disks. In case of squares, the objects have to stay axis-aligned during the reconfuguration. For an example, consider \cref{fig:intro}. Obstacles are modeled by defining the workspace in which
the objects are allowed to lie. At any time, the objects have to avoid the complement of the workspace.
Depending on the variant, the workspace is given as a polygon, a simple polygon, or a grid polygon,
that is an orthogonal polygon whose vertex coordinates are integers.
In unlabeled reconfiguration, the objects are indistinguishable, so that any object is allowed to end at any of the target positions. In labeled reconfiguration the target position for
every object is explicitly given.
\Cref{tab:results} presents an overview of the known results and the new results of this paper.

\begin{table}[htb]
    \centering
    \caption{Summary of results in motion planning for multiple objects}
\label{tab:results}
\begin{tabular}{@{}ccccc@{}}
        \toprule
        \textbf{Contribution} & \textbf{Labeled} & \textbf{Complexity} & \textbf{Objects} & \textbf{Workspace} \\ 
        \midrule
\cite{doi:10.1177/027836498400300405} & yes & $\PSPACE$-hard & rectangles & rectangle \\
\cite{DBLP:journals/tcs/HearnD05} & yes & $\PSPACE$-hard & $1\times 2$ rectangles & rectangle \\
\cite{DBLP:conf/fun/BrunnerCDHHSZ21} & yes & $\PSPACE$-hard & unit squares & polygon (with holes) \\
this work & yes & $\PSPACE$-hard & unit squares & simple polygon \\
\cite{kornhauser1984coordinating} & yes & $\P$ & unit squares & grid polygon \\
\cite{DBLP:journals/ijrr/SoloveyH16} & no & $\PSPACE$-hard & unit squares & polygon (with holes) \\
this work & no & $\PSPACE$-hard & unit squares & simple polygon \\
easy exercise & no & trivial `yes' & unit squares &  grid polygon \\
\hline
\cite{DBLP:journals/ipl/SpirakisY84} & yes & strongly $\NP$-hard & disks of three sizes & simple polygon \\
this work & yes & $\PSPACE$-hard & unit disks & polygon (with holes) \\
\cite{DBLP:conf/fun/BrockenHKLS21} & no & $\PSPACE$-hard & disks of two sizes & polygon (with holes) \\
this work & no & $\PSPACE$-hard & unit disks & polygon (with holes) \\
\bottomrule
\end{tabular}
\end{table}

Recent results of $\PSPACE$-hardness of geometric reconfiguration problems have been proven by a reduction from \textsc{Nondeterministic Constrained Logic} (NCL)~\cite{DBLP:conf/icalp/HearnD02,HearnDemaineBook}.
This holds for instance for the works~\cite{DBLP:conf/fun/BrockenHKLS21,DBLP:conf/fun/BrunnerCDHHSZ21,DBLP:journals/tcs/HearnD05,DBLP:journals/ijrr/SoloveyH16}.
Here, a polygon is built ``on top of'' an instance of NCL, which is in itself a plane graph (the problem NCL will be described in detail in \Cref{sec:mptsat}).
This technique inevitably leads to a polygon with holes, because the bounded faces of the NCL-instance are turned into holes of the polygon. 

A similar phenomenon appears in many proofs of $\NP$-hardness of two-dimensional geometric problems where the input is a polygon:
Such reductions are often from a variant of \PTSAT and work by building a polygon on top of the \PTSAT-instance we are reducing from.
Again, this leads to a polygon with a hole for every bounded face of the graph representing the \PTSAT-formula.

Abrahamsen and Stade~\cite{DBLP:journals/corr/abs-2404-09835} recently overcame this obstacle by making a completely different kind of reduction from \PTSAT in order to obtain $\NP$-hardness of the well-known problem of packing unit squares in a simple polygon; a problem that had been known to be $\NP$-hard for polygons with holes since  1981~\cite{DBLP:journals/ipl/FowlerPT81}.
Interestingly, Abrahamsen and Stade~\cite{DBLP:journals/corr/abs-2404-09835} used a non-planar geometric realization of a planar formula, where variables are represented as horizontal rows and edges are vertical columns.
These rows and columns are allowed to cross, but the planarity of the graph ensures that the endpoints of all rows and columns are incident to the outer face of the drawing.
This makes it possible to construct a simple polygon, following the boundary of the outer face, that ``implements'' the functionality of the rows and columns.
The rows and columns of the drawing are
represented by rows and columns of squares, and these likewise cross each other in the interior of the polygon.
The crucial observation is that movement in one direction does not influence movement in the other direction, so binary values can be ``transported''
independently in both directions through a crossing.

\subparagraph{Our contribution.}
In this paper we show $\PSPACE$-hardness of reconfiguration of unit squares in a simple polygon.
We thereby strengthen the result of Solovey and Halperin~\cite{DBLP:journals/ijrr/SoloveyH16}, who showed that the problem is $\PSPACE$-hard for polygons with holes.
To this end, we first show a result of independent interest. The solution space of boolean formulas
is given by an hypercube graph, where edges are formed between assignments which differ in exactly one variable. The satisfying assignments of a formula $\phi$ define a subgraph of this
graph, which we call $G(\phi)$. We call two satisfying assignments \emph{connected} if they are connected in $G(\phi)$. The reconfiguration problem for boolean formulas asks if two satisfying assignments are connected.
We show the following result.

\begin{restatable}{theorem}{monotoneSATpspace}\label{thm:monoton3SATpspace}
The reconfiguration problem for \MPTSAT formulas is $\PSPACE$-complete. 
\end{restatable}

For some other versions of  \TSAT, the reconfiguration problem has likewise been shown to be $\PSPACE$-hard, such as \textsc{Monotone-Planar-Not-All-Equal-3-SAT}~\cite{CARDINAL2020332}.
In Table~\ref{table:reconfiguration}, we give a summary of some variants of \TSAT and the complexity of their satisfyability and reconfiguration problems. 
Interestingly, most satisfiability and reconfiguration problems are hard, but there are some exceptions that can be solved in polynomial time.

 \begin{table}[hbt]
     \centering
     \caption{Overview of computational complexity for the satisfiability and reconfiguration problems of fundamental \TSAT variants.
}\label{table:reconfiguration}
     \begin{tabular}{l|c|c}
\toprule
& \textbf{Satisfiability} & \textbf{Reconfiguration} \\\hline
\TSAT & $\NP$-complete~\cite{Karp1972} & $\PSPACE$-complete \cite{dichotomy} \\
\PTSAT &$\NP$-complete \cite{doi:10.1137/0211025} & $\PSPACE$-complete \cite{CARDINAL2020332} \\
\textsc{Positive-1-in-3-SAT} &$\NP$-complete \cite{garey1979computers} & $\P$ (`no' iff assignments differ)   \\
\textsc{Positive-Not-All-Equal-3-SAT} &$\NP$-complete \cite{10.1145/800133.804350} & $\PSPACE$-complete~\cite{dichotomy}\\
\textsc{Mon.-Planar-Not-All-Eequal-3-SAT} & $\P$ \cite{10.1145/800133.804350} & $\PSPACE$-complete  \cite{CARDINAL2020332}
\\
\MPTSAT & $\NP$-complete~\cite{DBLP:journals/ijcga/BergK12} & $\PSPACE$-complete [This work]\\
\bottomrule
\end{tabular}

   \end{table}

Using \Cref{thm:monoton3SATpspace} we then show $\PSPACE$-hardness for reconfiguration of squares.
\begin{restatable}{theorem}{Squares}\label{thm:motionplanningNP}
Unlabeled reconfiguration of $8\times 8$ squares in simple grid polygons is $\PSPACE$-complete.
\end{restatable}

In our construction, the squares can only move locally. Hence, we exactly know the start and end position of each square and also obtain the same result for the labeled version. 
\begin{restatable}{corollary}{SquaresCor}\label{cor:squaresLabelled}
Labeled reconfiguration of $8\times 8$ squares in simple grid polygons is $\PSPACE$-complete.
\end{restatable}

\cref{thm:motionplanningNP,cor:squaresLabelled} strengthen the results of Solovey and Halperin~\cite{DBLP:journals/ijrr/SoloveyH16} and Brunner, Chung, Demaine, Hendrickson, Hesterberg, Suhl and Zeff~\cite{DBLP:conf/fun/BrunnerCDHHSZ21}, who showed that the problem is $\PSPACE$-hard for polygons with holes in the unlabeled and labeled variants, respectively. In contrast, for even simpler settings, the problem is trivial or polynomial decidable, e.g., it is an easy exercise to reconfigure any two configurations (with the same number) of unlabeled unit squares in grid polygons. The reconfiguration of unit squares in a rectangle includes the famous 15-puzzle and is generalizations which are polynomial time decidable~\cite{archer1999modern,kornhauser1984coordinating}.
In contrast, in  grid polygons, it is an easy exercise to show that reconfiguration  of unlabeled unit squares is always doable (if the number of squares agrees) and reconfiguration of labeled unit squares is polynomial time decidable~\cite{kornhauser1984coordinating}. Hence, our result is close to the border of polynomial-time decidability.

Lastly, we consider the reconfiguration of unit disks.

\begin{restatable}{theorem}{DisksUnlabeled}\label{thm:disks}
Reconfiguration of unlabeled unit disks in a polygon (possibly with holes) is $\PSPACE$-hard.
\end{restatable}

Before, it was only known that reconfiguration of unlabeled disks of two different sizes was $\PSPACE$-hard, as shown by Brocken, van der Heijden, Kostitsyna, Lo{-}Wong and Surtel~\cite{DBLP:conf/fun/BrockenHKLS21}. For disks of the same size so far only positive results where known under the assumption that there is sufficient separation between the start and target positions of the disks in a simple polygon~\cite{DBLP:journals/tase/AdlerBHS15,BanyassadyBBBFH22} but also in a polygon with holes~\cite{DBLP:conf/rss/SoloveyYZH15}. Thus, our results in particular show that the separation assumption in~\cite{DBLP:conf/rss/SoloveyYZH15} is necessary if $\P \neq \PSPACE$.

In fact, our constructed polygon guarantees that no two disks may swap their positions. Therefore, hardness for the labeled version follows immediately.
\begin{restatable}{corollary}{DisksLabeled}\label{cor:disksLabeled}
Reconfiguration of labeled unit disks in a polygon (possibly with holes) is $\PSPACE$-hard.
\end{restatable} 
Before, Spikaris and Yap~\cite{DBLP:journals/ipl/SpirakisY84} showed NP-hardness for labeld disks of three different sizes.

\subparagraph{Organization of the paper.}
The rest of the paper is structured as follows.
In \Cref{sec:mptsat}, we introduce NCL and use it to show that reconfiguration of \MPTSAT is $\PSPACE$-hard.
In \Cref{sec:unitsquares}, we then prove that the reconfiguration of unit squares in a simple polygon is $\PSPACE$-hard, by reducing from reconfiguration of \MPTSAT.
Finally, in \Cref{sec:unitdisks}, we show that the reconfiguration of unit disks in a polygon with holes is $\PSPACE$-hard, by a reduction from NCL.
We conclude with a discussion in \Cref{sec:discussion}.

\section{Reconfiguration of \MPTSAT}\label{sec:mptsat}

In this section, we show that the reconfiguration problem for \MPTSAT is $\PSPACE$-complete.

\monotoneSATpspace*

Containment in $\PSPACE$ follows from the fact that  \textsc{undirected $st$-connectivity} 
is in $\mathsf{LOGSPACE}$~\cite{Reingold} and the graph $G(\phi)$ of satisfying assignments is of size $2^n$ where $n$ is the number of variables in $\phi$.
To show $\PSPACE$-hardness we reduce from the 
configuration-to-configuration version of \textsc{Nondeterministic Constrained Logic} (NCL) \cite{ DBLP:conf/icalp/HearnD02,DBLP:journals/tcs/HearnD05,HearnDemaineBook}.

An NCL instance is given by an oriented planar cubic graph $G$ and another 
oriented graph $G'$, which is a reorientation of $G$.
The edges of $G$ are
colored blue or red such that every vertex is either incident to three blue edges,
or to one blue and two red edges. A feasible orientation must respect 
an indegree of at least two at every vertex, where blue edges count with multiplicity 2. 
An NCL-instance is positive, if and only if there is a sequence of edge-reversals that transforms $G$ to $G'$
while any intermediate orientation is feasible. For an example, consider \cref{fig:nclExample}.

\begin{figure}[htb]
    \centering
    \includegraphics[scale=1,page=16]{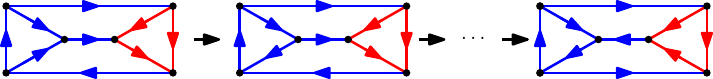}
    \caption{An instance of NCL. Can the left orientation be transformed
    to the right one by edge reversals such that all intermediate
    orientations are feasible?}
    \label{fig:nclExample}
\end{figure}

To prove \Cref{thm:monoton3SATpspace}, we first establish a correspondence between NCL instances and \PTSAT instances.
This correspondence already shows that reconfiguration of \PTSAT is $\PSPACE$-complete (which is already known, see discussion in~\cite{CARDINAL2020332}). 
Then, we show later how to adapt the construction for \MPTSAT.

\subsection {Constructing a planar {\sc \TSAT} formula}
\label{sec:make3sat}

We show that the reconfiguration problem of \MPTSAT is $\PSPACE$-complete, by describing how to create an instance of \MPTSAT from an instance $(G,G')$ of NCL.
The \MPTSAT instance consists of a formula $\phi$ and two assignments $\alpha$ and 
$\alpha'$ that satisfy $\phi$. The instance is positive if $\alpha$
and $\alpha'$ are connected in $G(\phi)$.

We first describe a mapping from $G$ to a planar \TSAT formula $\phi$ (equipped with a crossing-free embedding of the clause-variable-incidence graph in the plane) as follows.
We replace every vertex of $G$ by one or two {\em clauses} in the \TSAT formula, and every edge of $G$ by a {\em variable} in the \TSAT formula.
Specifically, for every oriented edge $e=uv$ in $G$, we introduce one variable $x_e$.
This variable
appears as $x_e$ at the clause node(s) induced by $v$
and as $\lnot x_e$ at the clause node(s) induced by $u$.
Then, a vertex of $G$ incident to three blue edges (also referred to as an {\em OR} vertex)  is replaced by a single clause for its three incident edges. A vertex of $G$ incident to a blue and two red edges (also referred to as an {\em AND} vertex) is replaced by two clauses, each for one of the red edges together with the blue edge, see \Cref{fig:3satreduction} for an illustration.

\begin{figure}[htb]
    \centering
    \includegraphics[scale=1]{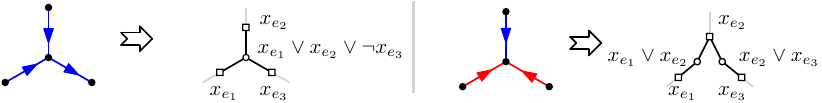}
    \caption{Mapping an NCL-Instance to an embedding of a planar \TSAT formula. Variable nodes are drawn as squares, clause
    nodes are drawn as disks.}
    \label{fig:3satreduction}
\end{figure}

The resulting \TSAT instance comes with a natural embedding based on the embedding of $G$. \Cref {fig:ncl-reduction} illustrates the construction for an example.

\begin{figure}[htb]
    \centering
    \begin{subfigure}{.25\textwidth}
        \centering
         \includegraphics[page=1]{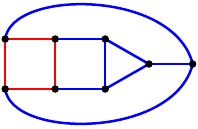}  
         (a)
    \end{subfigure}\hfill
    \begin{subfigure}{.25\textwidth}
        \centering
         \includegraphics[page=2]{figures/monotonePlanar3sat.pdf}  
         (b)
    \end{subfigure}\hfill
    \begin{subfigure}{.25\textwidth}
        \centering
         \includegraphics[page=7]{figures/monotonePlanar3sat.pdf}   
         (c)
    \end{subfigure}
    \caption{Graph $G$ of an NCL instance and corresponding embedded \TSAT formula $\phi$. (a) NCL instance $G$. (b) We place a variable node in the middle of each edge of $G$. (c) We replace each vertex of $G$ by one or two clause nodes, and connect them to their incident variable nodes appropriately.}
    \label{fig:ncl-reduction}
\end{figure}

The construction so far immediately provides us with a satisfying assignment for the \TSAT formula: if we set all variables to {\tt true}, then each clause must be satisfied because we started with a feasible orientation of $G$. For the {\em OR} vertices of $G$, this means at least one edge was oriented towards the vertex (and thus the corresponding clause has at least one positive literal). For the {\em AND} vertices of $G$, this means that either the blue edge or {\em both} red edges were oriented towards the vertex (and thus {\em both} corresponding clauses have at least one positive literal).

The satisfying assignments of $\phi$ and the feasible orientations of $G$ are in 1-1 correspondence.
The truth assignment of the formula encodes which edges of $G$ are reversed ({\tt true} variables correspond to edges that keep their initial orientation, and {\tt false} variables are the reversed edges). 
Thus, if $G_1$ and $G_2$ are two feasible orientations of $G$ that differ by only one
edge, then the associated satisfying truth assignments differ only in one variable.

\subsection {Monotone planar {\sc \TSAT}}

The construction from Section~\ref {sec:make3sat} results in a planar \TSAT instance, which is not necessarily {\em monotone}.
Our goal now is to create a planar monotone formula $\phi$ and two  truth assignments $\alpha$ and 
$\alpha'$. Planar monotone means that all literals in a clause should have the same sign (all positive or all negative).  In addition, there has to be a closed curve that passes through all variable nodes, such that it separates the positive from the negative clause nodes, and it does not cross any edges of the graph. 
We will show this in the following proof of \cref {thm:monoton3SATpspace}.

\begin{proof} [Proof of Theorem~\ref {thm:monoton3SATpspace}]
Consider an instance $(G, G')$ of NCL. For an illustration consider \cref{fig:ncl-reduction'}.  We begin with subdividing all edges of $G$ and obtain a new graph $\hat G$, i.e., we replace every edge $e=uv$ of $G$ by edges $uw$ and $wv$.
Note that although $\hat G$ looks the same as our layout for $G$ with additional edge vertices (in \cref {fig:ncl-reduction} (b)), the additional vertices of $\hat G$ have a different function.

\begin{figure}[htb]
    \centering
    \begin{subfigure}{.25\textwidth}
        \centering
         \includegraphics[page=8]{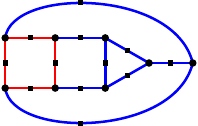}  
         (a)
    \end{subfigure}\hfill
    \begin{subfigure}{.25\textwidth}
        \centering
         \includegraphics[page=9]{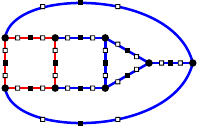} 
         (b)
    \end{subfigure}\hfill
    \begin{subfigure}{.25\textwidth}
        \centering
         \includegraphics[page=10]{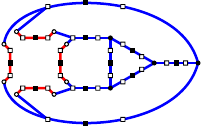}   
         (c)
    \end{subfigure}
    \caption{Constructing an embedded monotone \TSAT formula $\phi$ from a graph $G$ of an NCL instance. (a) The NCL instance $G$ is adapted to $\hat G$. (b) We place a variable node in the middle of each edge of $\hat G$. (c) We replace each vertex of $\hat G$ by one or two clause nodes, and connect them to their incident variable nodes. }
    \label{fig:ncl-reduction'}
\end{figure}

We now apply the construction from Section~\ref{sec:make3sat} to $\hat G$ rather than $G$: For every (directed)
edge $e=uv$ in $G$ we substitute
$u$ and $v$ in $\hat G$ as before, but the variable node at $u$ is now called $x_{e-}$
and the (now distinct) variable node at $v$ is now called $x_{e+}$.
The vertex of $\hat G$ that was introduced by subdividing $e$ will be associated with a 
clause node $(\lnot x_{e+} \lor \lnot x_{e-})$. In the clauses defined by
the original nodes of $G$ all variables are nonnegated (refer to Figure~\ref {fig:ncl-reduction'}).

A correspondence between satisfying truth assignment of $\phi$ and 
orientations of $G$ can be derived as follows:
Assume we have a feasible orientation $\tilde G$ of $G$.
If the orientation of an edge $e$ is
the same in $G$ and $\tilde G$, we set $x_{e-}=\tt{false}$ and $x_{e+}=\tt{true}$. 
Otherwise we set $x_{e-}=\tt{true}$ and $x_{e+}=\tt{false}$. This implies that
a variable $x_{e\pm}$ is ${\tt true}$ if in the derived orientation the edge $e$ points
towards the node that will be the clause-node of $x_{e\pm}$ with the three positive literals.
The assignment satisfies all positive clauses due to the reasons given in ~\Cref{sec:make3sat}.
For every clauses $(\lnot x_{e+} \lor \lnot x_{e-})$ ( coming from a subdivision point) exactly on variable is $\tt true$ and 
hence these clauses are also satisfied.
We define $\alpha$
with respect to $G$ and $\alpha'$ with respect to $G'$ according to these rules.

Using the above correspondence we can also obtain a feasible 
orientation of $G$ from a satisfying 
truth assignment of $\phi$ as follows.
If for an edge $e$ we have 
 $x_{e-}=\tt{false}$ and $x_{e+}=\tt{true}$ then we use for $e$ the same orientation
 as in $G$. If however, $x_{e+}=\tt{false}$ and $x_{e-}=\tt{true}$ then we use $e$ 
 in the reversed orientation of $G$. So far this represents our correspondence from the 
 previous paragraph.
The only other case can be  that $x_{e+}=\tt{false}$ and $x_{e-}=\tt{false}$
in a satisfying assignment for $\phi$. In this case the edge $e$ would not
contribute to the indegree constraint at neither $u$ nor $v$ in the reorientation of $G$. Thus,
we can orient $e$ arbitrarily, say in the original orientation. 

We argue now that the reduction is correct. Assume the we have a positive 
instance of NCL, in particular there is a sequence of edge reversals that 
will bring $G$ to $G'$ while staying feasible. We derive an assignment for $\phi$ for every
step as discussed above. If we reverse an edge $e$ we only need 
to exchange the values between
 $x_{e+}$ and $x_{e-}$. Assume  $x_{e+}={\tt true}$ and $x_{e-}={\tt false}$
 (the case 
 $x_{e+}={\tt false}$ and $x_{e-}={\tt true}$ is handled similarly). We will set first $x_{e+}= x_{e-} = {\tt false}$ first, and then  $x_{e+}={\tt false}$ and $x_{e-}={\tt true}$. 
 By this we only change one variable at a time while staying satisfying.
 Since the
 corresponding assignment of $G$ is $\alpha$ and the corresponding assignment of $G'$ is 
 $\alpha'$ we have that $\alpha$ and $\alpha'$ are connected in $G(\phi)$.

For the other direction assume that $\alpha$ and $\alpha'$ are connected in $G(\phi)$
by a sequence of satisfying assignments. Every step in this sequence changes one variable value.
Using the established correspondence this yields the reversal of a single edge or no 
edge is reversed. Thus, the sequence of corresponding assignments brings $G$ to $G'$
through a sequence of edge reversals and all intermediate orientations are feasible.

\begin{figure}[htb]
    \centering
    \begin{subfigure}{.25\textwidth}
        \centering
         \includegraphics[page=1]{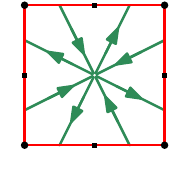}  
         (a)
    \end{subfigure}\hfill
    \begin{subfigure}{.25\textwidth}
        \centering
         \includegraphics[page=2]{figures/euler.pdf}  
         (b)
    \end{subfigure}\hfill
        \begin{subfigure}{.25\textwidth}
        \centering
         \includegraphics[page=3]{figures/euler.pdf}  
       (c)
    \end{subfigure}\hfill
        \begin{subfigure}{.25\textwidth}
        \centering
         \includegraphics[page=4]{figures/euler.pdf}  
        (d)
    \end{subfigure}\hfill
    \caption{Obtaining a suitable closed curve intersecting all edges in $G'$ -- local situation around a vertex. (a) The
    initial orientation. (b) Decomposition into closed curves by shortcutting. (c--d) Gluing curves together.}
    \label{fig:euler}
\end{figure}

Finally, we have to argue that there exists a closed curve passing through the variable vertices 
that separates the positive from the negative clause nodes.
Consider the dual graph of $\hat G$ according to
the embedding of $\phi$.
Let $uv$ be an edge of $\hat G$ with $v$ be a vertex from $G$ and $u$ a vertex that was introduced by subdividing an 
edge of $G$. Note that in our construction $u$ will be replaced
with an all-negative clause node and $v$ with an all-positive clause node.
Let the dual edge of $uv$ be $xy$. We direct
$xy$ such that when traversing from $x$ to $y$ we cross
$uv$ having $u$ on the left and $v$ on the right; see Fig.~\ref{fig:euler}(a).
By construction the vertices in the dual graph have
alternating in- and outgoing edges. We can decompose this graph by a set 
of oriented cycles by \enquote{linking} every ingoing edge at a vertex with its clockwise
following (outgoing) edge. By short-cutting these cycles every time they pass a vertex,
we get a set of disjoint closed curves inheriting the orientation of the dual edges; see Fig.~\ref{fig:euler}(b). When two
such curves meet around a vertex, we can glue them together to one curve respecting the orientation; see Fig.~\ref{fig:euler}(c--d). 
By repeating this
process we obtain a single curve.  In the \MPTSAT instance we can route the curve through 
the variable nodes whenever it would cross an edge in $\hat G$. By our choice of the
orientation of the dual graph
all-positive clause nodes will lie on the right 
and the all-negative clause nodes on the left side of the curve. \cref{fig:example} illustrates the construction of the curve for an example.
\begin{figure}[htb]
    \centering
    \begin{subfigure}{.3\textwidth}
        \centering
        \includegraphics[page=13]{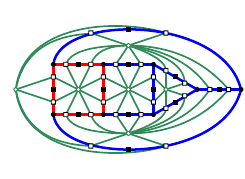}  (a)
    \end{subfigure}\hfill
    \begin{subfigure}{.3\textwidth}
        \centering
         \includegraphics[page=14]{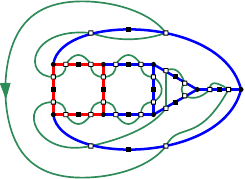}  (b)
    \end{subfigure}\hfill
    \begin{subfigure}{.3\textwidth}
        \centering
         \includegraphics[page=15]{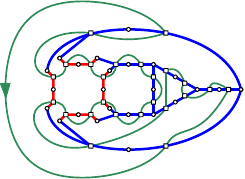}   (c)
    \end{subfigure}
    \caption{Finding a separating closed curve for the instance of \MPTSAT by constructing a Euler tour in the dual graph of $G'$. (a) The dual graph of $G'$. (b) A Euler tour in the dual graph. (c) The curve following the Euler tour in the final embedded \MPTSAT instance.}
    \label{fig:example}
\end{figure}
\end{proof}

\section{Unit squares in simple polygons}\label{sec:unitsquares}

In this section, we prove \Cref{thm:motionplanningNP}
by reducing \MPTSAT to the 
unlabeled reconfiguration of squares in simple grid polygons.
We start with a \MPTSAT formula $\phi$ and two satisfying assignments $\sigma_0$ and $\sigma_1$ of $\phi$.
In our reduction, we create a grid polygon $P$ and use the assignments $\sigma_0$ and $\sigma_1$ to define two configurations $c_0$ and $c_1$ of $8\times 8$ squares in $P$. 
We then show that there is a reconfiguration from $c_0$ to $c_1$ in $P$ if and only if there is a reconfiguration from $\sigma_0$ to $\sigma_1$ in~$\phi$.
The squares are axis-aligned and may move by translations only.

\subsection{Reference centers}

As a preliminary step we provide a helpful tool for verifying the correctness of the reduction.
For every $y\in\mathbb Z$, we fix a number $b_y\in [0,8)$.
Consider a grid polygon $P$ with the property that for all $x,y\in\mathbb Z$, it either holds that the square $S_{xy}=[8x+b_y,8x+b_y+4]\times [8y,8y+4]$ is contained in $P$ or that $S_{xy}$ is interior-disjoint with $P$, see also \cref{fig:refcenter}.
When there exists such values $b_y\in [0,8)$, we say that $P$ has \emph{reference centers}, and the squares $S_{xy}$ contained in $P$ are the \emph{reference centers} of $P$.
Moreover, we call a configuration  of $8\times 8$ squares in $P$  \emph{perfect} if the number of squares and reference centers match.

\begin{figure}[htb]
    \centering
    \includegraphics[page=5,width=0.5\linewidth]{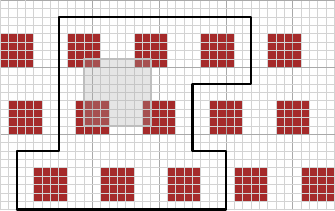}
    \caption{A grid polygon that has reference centers  for $b_0=1, b_1=6, b_2=5$ and an $8\times 8$ square intersecting four reference centers; the area of intersection is 16.}
    \label{fig:refcenter}
\end{figure}

\begin{lemma}\label{lem:area4}
Let $P$ be a grid polygon that has reference centers and consider an $8\times 8$ square~$S$ placed in $P$.
Then the area of the intersection of $S$ with the reference centers is $16$.
\end{lemma}
\begin{proof}
Assume that the grid cells of the reference centers are colored red. 
Consider every infinite horizontal line $L_i=[-\infty,i] \times [\infty,i+1]$ of grid cells.
We call $L_i$ \emph{red} if it contains some red cells. When going through theses lines
in increasing order four red lines will be always followed by 4 non-red lines and vice versa. 
Thus, every $8\times 8$ square in $P$ intersects exactly 4 red lines. 
Note that any sequence of 8 horizontally consecutive cells in a red line
contains exactly 4 red cells. Therefore, every $8\times 8$ square intersects exactly 16
red cells.
\end{proof}

The next lemma allows us to assign one square in a perfect configuration to each reference center.

\begin{lemma}
\label{lem:refCenters}
Let $P$ be a grid polygon (possibly with holes) that has reference centers and consider a perfect configuration of $8\times 8$ squares in $P$.
Then each square in the configuration contains one reference center and is interior-disjoint from all other reference centers.
\end{lemma}

\begin{proof}
We proceed by induction on $k$, the number of reference centers.
In the base case $k=1$, there is one square in the perfect configuration, and it contains the unique reference center by \Cref{lem:area4}.

For the induction step, suppose that the claim holds for polygons with at most $k$ reference centers Let $P$ be a grid polygon with $k+1$ reference centers and consider a perfect configuration in $P$.
Because the reference centers have a total area of $16(k+1)$ and the area of reference centers covered by each square in the configuration is $16$ by \Cref{lem:area4}, together the squares  cover all reference centers.
Let $S$ be the square in the configuration with a center where the sum of the $x$- and $y$-coordinates is minimum.
If $S$ does not contain a reference center, then an edge of $S$ cuts through the interior of a reference center by \Cref{lem:area4}, and the remaining part of the reference center must be covered by another square.
This cannot be the left or bottom edge of $S$, as that would contradict the minimality of $S$.
But if the top or right edge cuts through a reference center, it follows that the bottom or left edge, respectively, also does, due to the pattern in which the reference centers are placed.
Hence, $S$ contains a reference center and is interior-disjoint from the other reference centers by \Cref{lem:area4}.

The set $P'=P\setminus S$ consists of  one or more polygons with $k$ reference centers in total.
Note that the removal of $S$ might also yield a hole in $P'$.
The remaining squares (after removing $S$) constitute perfect configurations in the polygons of $P'$. Hence, by induction, each square contains a reference center and is interior-disjoint from the other reference centers.
Therefore, the claim holds for all squares in the configuration of $P$.
\end{proof}

By \cref{lem:refCenters}, in any reconfiguration from one perfect configuration to another  each square contains the same reference center the whole time, i.e., the squares  merely shift slightly sideways and up and down around the same reference center.
It turns out that in order to describe the mechanics of the intended reconfigurations of the squares in $P$ that correspond to reconfigurations of the formula $\phi$, we only need to consider the positions with respect to their
reference centers that are either extreme or centered in each direction.
This leaves us with the nine possible positions shown in \Cref{fig:colorcodesA}.
In \Cref{sec:verification}, we prove that also arbitrary continuous reconfigurations in $P$ correspond to reconfigurations of $\phi$, but for understanding the key ideas of the reduction, it suffices to consider these nine discrete positions.

\begin{figure}[htb]
\centering   
    \begin{subfigure}{.26\textwidth}
        \centering
         \includegraphics[page=19]{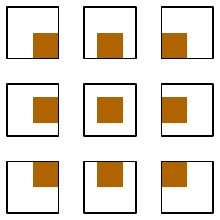}   
        \subcaption{}
        \label{fig:colorcodesA}
    \end{subfigure}\hfil
        \begin{subfigure}{.26\textwidth}
        \centering
         \includegraphics[page=20]{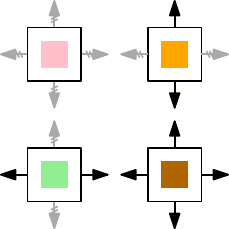}   
       \subcaption{}
       \label{fig:colorcodesB}
    \end{subfigure}
\caption{(a) The up to nine important placements of a square for a given reference center.\\
(b) Examples of colored reference centers and how the squares containing them may move.} 
\label{fig:colorcodes}
\end{figure}

When the top (resp.~bottom) edge of a square aligns with that of its reference center, we say that the square is \emph{down} (resp.~\emph{up}).
Otherwise, the midpoints of the reference center and the square have the same $y$-coordinate, in which case we say that the square is \emph{vemid} (abbreviation of vertically in the middle).
When the left (resp.~right) edge of the square aligns with that of the reference center, the square is \emph{right} (resp.~\emph{left}).
Otherwise, the midpoints of the reference center and the square have the same $x$-coordinate, in which case we say that the square is \emph{homid} (abbreviation of horizontally in the middle).

For most squares, several of the nine possible positions can be excluded, for instance due to being pushed by the boundary of $P$ or another square.
In our figures, we use a color scheme for the reference centers, to make the reader aware of what possible positions a square has.
The colors have the following meanings: the square can not move (\emph{pink}),
only move vertically (\emph{orange}),
 only move horizontally (\emph{green}),
 move vertically and horizontally (\emph{brown}). 
There might exist additional restrictions on the positions that are not indicated by the colors.

\subsection{Schematics of the construction}
On a big scale, our construction follows the ideas of Abrahamsen and Stade~\cite{DBLP:journals/corr/abs-2404-09835}. Therefore, we start by reviewing the high level ideas.
Let $\phi$ be an instance of \MPTSAT\ with variables $x_1,\ldots,x_k$ and consider the variable-clause-incidence graph $G$ with a suitable embedding.
We now describe schematically the construction of a grid polygon $P$.
Each feasible solution to $\phi$ corresponds to a configuration of $8\times 8$ squares in $P$, and there is a reconfiguration between two reconfigurations if and only if there is a reconfiguration between the two solutions of $\phi$.

In the first step, we introduce a \emph{\segRepr} that consists of horizontal segments for vertices and clauses and of vertical segments for the edges of $G$, see in \Cref{fig:schematicsB}. The idea behind the \segRepr is as follows: The horizontal segments will correspond to one or several rows of squares which transmit information horizontally. Similarly, the vertical segments will correspond to columns of squares that transmit information vertically. Moreover, the final polygon roughly corresponds to the smallest enclosing orthogonal polygon.

\begin{figure}[htb]
\centering
    \begin{subfigure}{.15\textwidth}
        \centering
         \includegraphics[page=21,scale=.85]{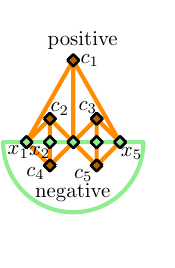}
         \subcaption{}
         \label{fig:schematicsA}
    \end{subfigure}\hfil
    \begin{subfigure}{.3\textwidth}
        \centering
         \includegraphics[page=22,scale=.85]{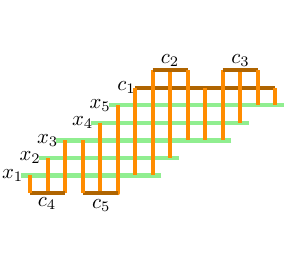}
         \subcaption{}
         \label{fig:schematicsB}
    \end{subfigure}\hfil
       \begin{subfigure}{.5\textwidth}
        \centering
         \includegraphics[page=23,scale=.8]{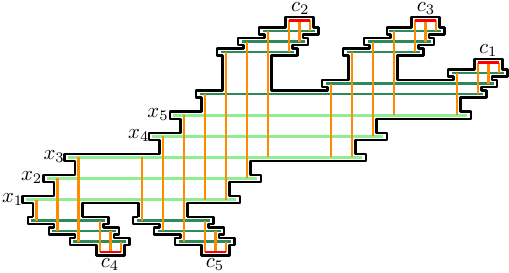}
         \subcaption{}
         \label{fig:schematicsC}
    \end{subfigure}
\caption{Schematic construction for transforming an instance of \MPTSAT to a polygon. (a) a variable-clause incidence graph. (b) first step for constructing the \segRepr of $G$ (c) second step for constructing the \segRepr of $G$: introducing the auxiliary variable rows and OR gadgets.}
\label{fig:schematics}
\end{figure}

To construct the \segRepr, we introduce a horizontal segment for each variable $x_1,\ldots,x_k$ in this order bottom-up, where the $x$-coordinates of the left and right endpoints are increasing, see \cref{fig:schematics}.
Furthermore, the right endpoint of $x_1$ is to the right of the left endpoint of $x_k$, so that all the segments have a common horizontal overlap.
These are the \emph{primary} variable rows.
The positive clauses are represented by 
horizontal segments above $x_k$ and the negative clauses by horizontal segments  below $x_1$. We say that a clause $c'$
is \emph{nested} in a clause $c$, if in 
$G$ the clause node of $c'$ lies inside a cycle defined
by two incident edges of the node of $c$ closed with a path between variable nodes. 
For instance, in \Cref{fig:schematicsA}, $c_2$ and $c_3$ are nested inside $c_1$. Note that the property of being nested 
defines a partial order.
If a positive clause $c_i$ is nested inside $c_j$, then we place the segment of $c_i$ above that of $c_j$ in the \segRepr.
If instead the clauses are negative, the segment of $c_i$ is below that of $c_j$.
Each edge $x_ic_j$ of $G$ is realized by a vertical segment connecting the segments of $x_i$ and $c_j$.
All vertical segments to the positive clauses are placed to the right of all those to the negative clauses.
The left-to-right ordering of the edges to the positive clauses along the variables $x_1,\ldots,x_k$ in $G$ is preserved by the corresponding vertical segments, as well as the ordering of the edges to the negative clauses.
The segment of a positive (resp.~negative) clause $c_j$ starts at the top (resp.~bottom) endpoint of the vertical segment corresponding to the leftmost edge incident to $c_j$ in $G$ and ends at the top (resp.~bottom) endpoint of the rightmost edge.

In the next step, illustrated in \Cref{fig:schematicsC}, we replace each clause segment by three \emph{auxiliary} variable segment, one for each of the variables connected to the clause.
In the right side of the auxiliary segments, we connect them using vertical segments to an OR gadget, represented by a red horizontal segment.
As sketched in the figure, we construct (guided by the partial nested-order of the positive and negative clauses) a polygon whose boundary approximately follows the outer face of the resulting representation, but the finer details will be given later.
The following lemma was proved in~\cite{DBLP:journals/corr/abs-2404-09835}.

\begin{lemma}
There is a \segRepr as described, where the segments representing the OR gadgets are not crossed by any vertical segments, and all segment endpoints are incident to the outer face of the drawing.
\end{lemma}

Algebraically, the introduction of an auxiliary variable row for a clause $c_t$ corresponds to the following equivalence:
\[
 (x_i\vee x_j \vee x_k)
\Leftrightarrow
 \Big(\exists y_{t,i},y_{t,j},y_{t,k}:
(y_{t,i} \Rightarrow x_i)\wedge (y_{t,j}\Rightarrow x_j)\wedge(y_{t,k}\Rightarrow x_k)\wedge (y_{t,i}\vee y_{t,j}\vee y_{t,k})\Big)
\]
Here, $x_i,x_j,x_k$ are the primary (original) variables and $y_{t,i},y_{t,j},y_{t,k}$ are the auxiliary variables introduced for that particular clause. We remark that the formula is not monotone anymore, but the only non-monotone clauses are implications.
Note that the equivalence also preserves reconfiguration:
If, say, $x_i$ changes to {\tt false}, then we change $y_{t,i}$ to {\tt false} before changing $x_i$.
If $x_i$ changes to {\tt true}, then we change $y_{t,i}$ to {\tt true} after $x_i$.
For a negative clause $c_t$ , we use an analogous equivalence:
\begin{align*}
& (\neg x_i\vee \neg x_j \vee \neg x_k)
\Leftrightarrow \\
& \Big(\exists y_{t,i},y_{t,j},y_{t,k}:
(\neg y_{t,i}\Rightarrow \neg x_i)\wedge (\neg y_{t,j}\Rightarrow \neg x_j)\wedge(\neg y_{t,k}\Rightarrow \neg x_k)\wedge (\neg y_{t,i}\vee \neg y_{t,j}\vee \neg y_{t,k})\Big)
\end{align*}

\subsection{Components and gadgets}

We informally distinguish between \emph{gadgets} and \emph{components} in our constructions, where a gadget is given by a set of segments on the boundary creating some simple functionality, and a component is thought of as a whole region of the polygon that can involve an arbitrary number of gadgets.

On a very high level the truth assignment of a variable of $\phi$ is encoded by a configuration of a so-called
\emph{variable component}. The variable component (used to represent primary and auxiliary variables) has essentially three possible states, namely \emph{plus}, \emph{minus} and \emph{transitional}.
The first two correspond to the values {\tt true} and {\tt false} of a binary variable, while the truth value of the transitional state is defined to be the {\tt true}/{\tt false} value of the previous plus/minus state. In order 
to split the information represented in a variable component of a primary variable and transfer it to a variable component of an auxiliary variable we use PUSH gadgets. In particular, these gadgets are used to model the implication-clauses and they come in two versions:  PUSH-UP-IF-MINUS and  PUSH-DOWN-IF-PLUS. The former is used to model implications of the form $y \Rightarrow x$, the later for $\lnot y \Rightarrow \lnot x$. 
The gadgets use a vertically moving stack of square rows with growing width (called pyramid) to transmit information from primary to auxiliary variables, see for example \cref{fig:RowAlignment}. For a PUSH-UP-IF-MINUS gadget representing an implication $y \Rightarrow x$ a pyramid that is raised enforces 
that the auxiliary variable component of $y$ is set to minus, if the variable component of $x$ is set to minus. Similarly, 
a PUSH-DOWN-IF-PLUS gadget modeling $\lnot y \Rightarrow \lnot x$ has the functionality that a pyramid that is lowered enforces 
that the auxiliary variable component of $y$ is set to plus, if the variable component of $x$ is set to plus. 
We will make sure that the vertically moving pyramid structure do not collide with the vertically operating variable components
they \enquote{cross}.
To only allow satisfying assignments we place for every clause an \emph{OR gadget} that is connected with the corresponding
auxiliary variable components. We reuse the ideas for the PUSH Gadgets. For a clause $(y_i\lor y_j \lor y_k)$ we 
the auxiliary variables will push a pyramid up into the clause gadget if the corresponding auxiliary variable is set to minus.
The OR gadget has only place for two of these pyramids to push up. The clauses $(\lnot y_i\lor \lnot y_j \lor \lnot y_k)$
are handled symmetrically. Here, an auxiliary variable set to plus would result in a pyramid pushing down and a the OR
gadget can only incorporate two lowered pyramids.

We continue with describing the gadgets and components in detail.

\subsection{The variable component}\label{sec:simplevar}

The variable component has two \emph{main} rows of squares and an arbitrary number of PUSH gadgets, ensuring dependency between the primary and auxiliary variables and between the auxiliary variables and the OR gadgets.
The PUSH gadgets require the introduction of \emph{helper rows} padded along the main rows, above the top row or below the bottom row.
The construction of PUSH gadgets and helper rows will be explained in later sections.

The construction of the main rows of a variable component is shown in \Cref{fig:variable}.
A perfect configuration of the component has three possible states, as defined by the horizontal position of the two rightmost squares (covering the brown reference centers).
Note that none of the two squares can push right, and if one is homid, the other needs to be right. 
If the bottom is in the homid position, the state is \emph{plus}.
If the top is in the homid position, the state is \emph{minus}.
Otherwise, the state is \emph{transitional}.

\begin{figure}[htb]
    \centering
    \begin{subfigure}{.27\textwidth}
        \centering
         \includegraphics[page=15]{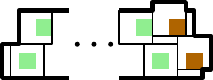}  
         \subcaption{plus state}
    \end{subfigure}\hfil
    \begin{subfigure}{.27\textwidth}
        \centering
         \includegraphics[page=16]{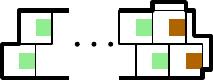}  
         \subcaption{transitional state}
    \end{subfigure}\hfil
    \begin{subfigure}{.27\textwidth}
        \centering
         \includegraphics[page=17]{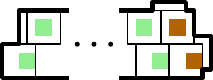}  
         \subcaption{minus state}
    \end{subfigure}
    \caption{A variable component and its three states, defined by the horizontal positions of the two rightmost squares with brown reference centers.}
\label{fig:variable}
\end{figure}

Now we describe what happens when a block of raised or lowered squares \enquote{crosses} a variable component.
As discussed before, pyramids made of such blocks are used to create dependencies between two variable components, but it is important that they can cross other variable components without interacting with them.

When a stack of raised squares crosses a row of squares with a different horizontal alignment, the stack of raised squares will grow in width.
Hence, the stack will be pyramid shaped if all rows have different alignment than the neighboring rows.
It is important that we know the width of a pyramid when it reaches the target variable component where it will end.
Thus, we need to know exactly how many times the width is increased and which squares are raised in the top row.
To this end, and in order to create a necessary amount of vertical spacing, we introduce \emph{static} rows in between all neighboring pairs of variable components.
A static row consist of squares in the beginning and in the end which cannot move; some squares in the middle can move vertically up and down, see \Cref{fig:static}. 
\begin{figure}[htb]
\centering
\includegraphics[page=3]{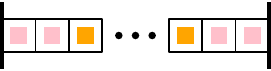}
\caption{A static row.
The squares have fixed horizontal positions, but some of the middle squares will be able to move up or down because of pyramids.}
\label{fig:static}
\end{figure}
So in particular, all squares in a static row have a fixed horizontal position. 

Moreover, we ensure that the fixed horizontal position is never aligned with the rows of the variable component.
Furthermore, the top row of a variable components is never aligned with the bottom row, see \Cref{fig:RowAlignment}.
Hence, when a pyramid crosses a variable component, the specific horizontal positions of the rows do not affect which squares are raised/lowered on the other side of the variable.
This is expressed in the following lemma.

\begin{figure}[htb]
\centering
\includegraphics[page=9]{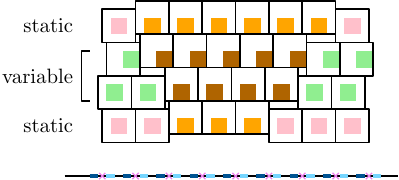}
\caption{A pyramid crossing a variable component.
The scale at the bottom shows the intervals of the possible $x$-coordinates of the left corners of the upper and lower squares in the variable rows (light and dark blue, respectively) and those of the squares in the static rows (pink).
Since these three groups of intervals are pairwise disjoint, the pyramid grows in width by one square per row.
Hence, the width of the pyramid always grows by three squares when crossing a variable component, regardless of the state of the variable component.
}
\label{fig:RowAlignment}
\end{figure}

\begin{lemma}\label{lem:PyramidGrowth}
Consider a perfect configuration.
If a block of consecutive squares in a static row below (above) a variable component are up (down), then the corresponding squares in the static row above (below) the variable component must also be up (down) as well as two extra squares on the right (left) and one on the left (right).
\end{lemma}

\begin{proof}
The case where a block of squares are up is illustrated in \Cref{fig:RowAlignment}. The horizontal alignment of squares in the static rows is fixed by the polygon edges as shown in \Cref{fig:static}. The squares in each row of the variable gadget must each be in the homid or left, as shown in \Cref{fig:variable}. So the squares in the lower variable row that are up in \Cref{fig:RowAlignment} always overlap with the block of static squares, so they must likewise be up. Similarly, the squares that are shown to be up in the upper variable row and upper static row need to be up.
The case where a block of squares are down is analogous.
\end{proof}

Recall that there may be helper rows above or below the main rows of a variable component, which will be explained in more detail in \Cref{sec:implies}.
Conceivably, these could be unaligned with the main rows, which would cause a pyramid to grow more in width when crossing the variable.
However, if they are aligned (which is our ``intended behavior'' of the helper rows), then the width will not grow beyond what is expressed by \Cref{lem:PyramidGrowth}.

\subsection{PUSH gadgets and helper rows}\label{sec:implies}

\Cref{fig:puim,fig:pdip} show the PUSH-UP-IF-MINUS and PUSH-DOWN-IF-PLUS gadgets, respectively.
The gadgets are mirror images of each other by a horizontal axis.
In order to make PUSH gadgets, we need a square that is pushed up or down only when it is left.
To this end, we introduce \emph{helper rows}, which are extra rows above and below the main variable rows.
Consider the PUSH-UP-IF-MINUS gadget, shown in \Cref{fig:puim}.
If the squares of the lower main row are left, the squares of the helper row will be left as well, because the square $a$ presses down $b$, which thus pushes $c$ to the left.
Therefore the square $d$ will be up, and a pyramid of raised squares is created.

\begin{figure}[htbp]
\centering
\includegraphics[page=4, width=\textwidth]{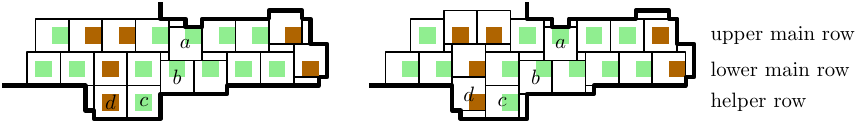}
\caption{The plus and minus states of PUSH-UP-IF-MINUS.
In minus, square $d$ is pushed up.}
\label{fig:puim}
\end{figure}

The PUSH-DOWN-IF-PLUS gadget works by analogous mechanics. If the squares in the upper main row are left, two squares of the lower row are pushed down, see \Cref{fig:pdip}. 

\begin{figure}[htbp]
\centering
\includegraphics[page=5]{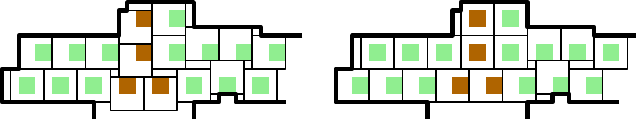}
\caption{The plus and minus states of PUSH-DOWN-IF-PLUS. In plus, the squares with brown reference centers are pushed down.}
\label{fig:pdip}
\end{figure}

All helper rows are created at the right end of the variable gadget and end at a PUSH gadget.
\Cref{fig:helperschematics} shows schematically where the helper rows are placed, and \Cref{fig:helperrows} shows a detailed example with multiple helper rows.
Note that with multiple helper rows, the same principle is used as with a single helper row in order to ensure that when the lower main row is left, then all helper rows below are also left, so they all create pyramids of raised squares.

\begin{figure}[htbp]
\centering
\includegraphics[page=10,scale=.8]{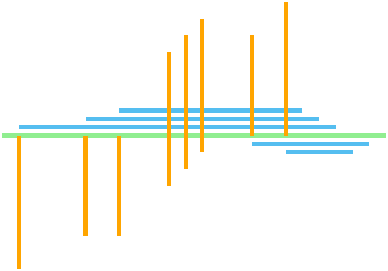}
\caption{Schematic showing the horizontal span of the helper rows (blue) along a principal variable component. 
The variable connects to three auxiliary variables below and two above, so there are three helper rows above and two below.
In between the two groups, pyramids from other principal variable components can cross, indicated as three orange vertical segments in the middle.}
\label{fig:helperschematics}
\end{figure}

\begin{figure}[htbp]
\centering
\includegraphics[width=\textwidth,page=6]{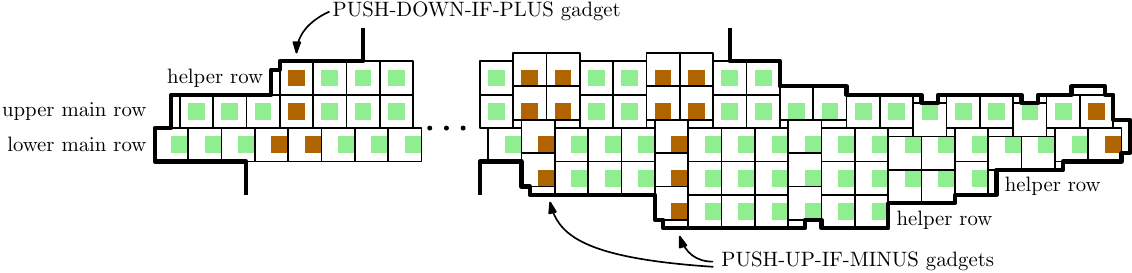}
\caption{A variable component in the minus position with two PUSH-UP-IF-MINUS gadgets and one PUSH-DOWN-IF-PLUS gadget.
Thus, there are two helper rows at the bottom and one at the top.
Other pyramids may cross the variable in the region indicated by the dots.}
\label{fig:helperrows}
\end{figure}

\begin{figure}[tb]
\centering
\includegraphics[width=\textwidth,page=7]{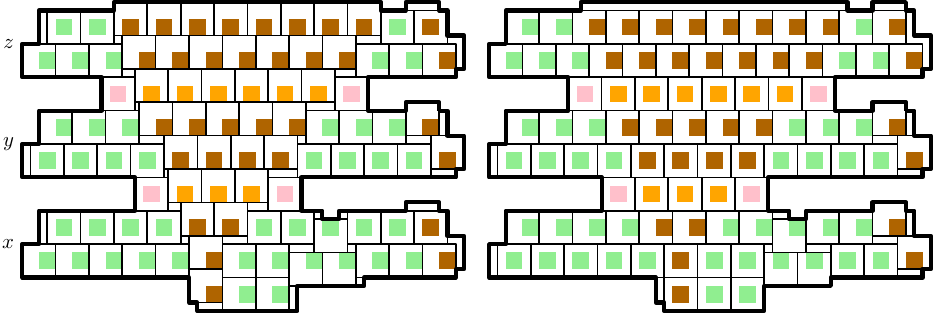}
\caption{There is a PUSH-UP-IF-MINUS gadget on the bottom variable $x$, which connects $x$ to the top variable $z$.
When $x$ is minus, a pyramid of raised squares is created.
The boundary of the polygon at $z$ forces $z$ to also be minus.
The state of $y$ has no influence.
When $x$ is plus, the pyramid does not have to raise, so $z$ is not restricted.
}
\label{fig:implication}
\end{figure}

We now describe how a pyramid ends at an auxiliary variable component; see \Cref{fig:implication} for a concrete example to support the description.
When we make a PUSH-UP-IF-MINUS gadget from a primary variable~$x$ to an auxiliary variable $y$, we create a pocket in the boundary of $P$ along the upper primary row of $y$ to enable the squares of the pyramid to raise.
We place the pocket so that the squares can only raise if they are homid (and if the pyramid has grown in width by exactly three squares every time it crossed another variable component).
Thus, the rightmost square of the top row of $y$ is also homid, so $y$ must be minus.
We have thus made the implication $y \Rightarrow x$.

In minus, pyramids are created to the positive auxiliary variables, which in turn create pyramids to the positive clauses.
Likewise, in plus, pyramids are created to the negative clauses.
In transitional, pyramids are created to all clauses.
As we will see in \Cref{sec:or}, there will only be room for two pyramids for each clause.
Hence, each clause ensures that one of the variables is in the plus/minus state that represents the {\tt true}/{\tt false} value that satisfies the clause.
We summarize our findings as the following lemma.

\begin{lemma}\label{lem:implication}
In every perfect configuration, the following holds.
If a variable is not plus (minus), a pyramid of raised squares appears from any PUSH-UP-IF-MINUS (PUSH-DOWN-IF-PLUS) gadget, and if the pyramid is received by an auxiliary variable, that variable must be minus (plus).
As a consequence, we can realize the implications $y \Rightarrow x$ and $\lnot y \Rightarrow \lnot x$, i.e., ensure that the values of variables encoded by the configuration satisfy the implications.
\end{lemma}

\subsection{OR gadget}\label{sec:or}

\begin{figure}[htbp]
\centering
\includegraphics[page=12]{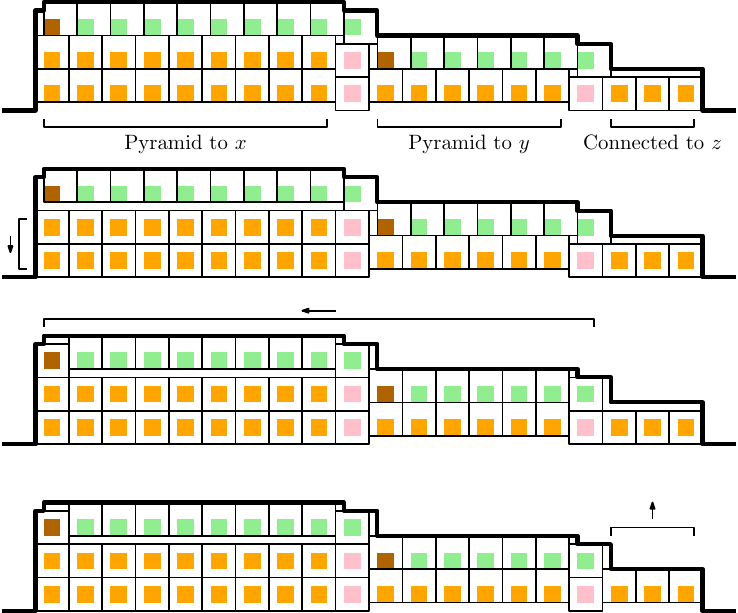}
\caption{The OR gadget for a clause $x\lor y\lor z$.
Initially, there is a pyramid from $x$ and $y$.
The pyramid from $x$ is removed, and it is shown how the top squares reconfigure so that a pyramid from $z$ can raise.
}
\label{fig:orgadget}
\end{figure}

\Cref{fig:orgadget} shows the positive OR gadget.
The gadget is connected to three auxiliary variable components.
Using PUSH-UP-IF-MINUS gadgets, pyramids can raise from the variables and end in the OR gadget.
The gadget allows for two pyramids to be raised, but not all three at once.
Hence, one variable must be plus.
When at most one pyramid is raised, the squares in the top of the gadget can reconfigure so that there is room for any of the other two pyramids to raise.
An example is shown in the figure---the reader can easily check that it is likewise possible in the other five cases where one pyramid is raised and we need to make room for one of the other two to raise as well.
This corresponds to a reconfiguration of the formula $\phi$:
If $x\lor y\lor z$ is a clause and one of the variables, say $z$, changes from positive to negative, then $x$ or $y$ must have been positive before $z$ changes (otherwise the clause would not be satisfied after $z$ changing).
Hence, there was only one pyramid, so the squares in the gadget can reconfigure to accommodate for the pyramid from $z$ to raise.

\Cref{fig:orgadgetinstallation} shows how the OR gadget is connected to the variable components.
The negative OR gadget is obtained by reflection along a horizontal axis.

\begin{lemma}\label{lem:or}
In a positive OR gadget, at least one of the connected auxiliary variables is plus. In a negative OR gadget, at least one of the connected auxiliary variables is minus.
\end{lemma}

\begin{figure}[htbp]
\centering
\includegraphics[page=11,width=\textwidth]{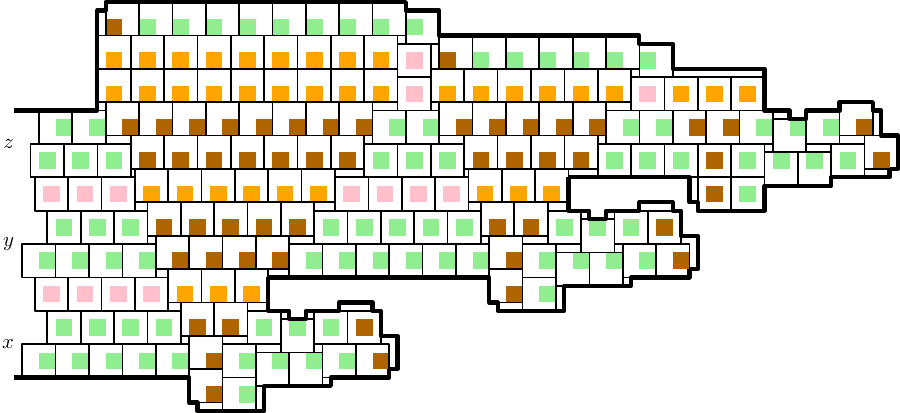}
\caption{The positive OR gadget connected to the auxiliary variable components for $x$, $y$ and $z$.
Here, $x$ and $y$ are minus, so they create pyramids, but $z$ is plus, which allows for a configuration.}
\label{fig:orgadgetinstallation}
\end{figure}

\subsection{Verification}\label{sec:verification}

Recall that we start with a \MPTSAT formula $\phi$ and two satisfying assignments $\sigma_0$ and $\sigma_1$ of $\phi$.
In our reduction, we create a polygon $P$ as described.
We define start and target configurations $c_0$ and $c_1$ according to $\sigma_0$ and $\sigma_1$:
We place squares in all primary and auxiliary variable components in $c_i$ in the plus/minus state corresponding to the value of the variable in $\sigma_i$.
There is room for the pyramids created by PUSH gadgets by construction.
We now need to verify that there is a reconfiguration between $\sigma_0$ and $\sigma_1$ in $\phi$ if and only if there is one between $c_0$ and $c_1$ in $P$.
This is expressed by the following two lemmas.

\begin{lemma}\label{lem:dir1}
If $\phi$ has a reconfiguration between two satisfying assignments $\sigma_0$ and $\sigma_1$, then there is also a reconfiguration between the corresponding start and target configurations $c_0$ and $c_1$ in $P$.
\end{lemma}

\begin{proof}
We consider the event that the variable $x$ changes from positive to negative in the reconfiguration of the formula $\phi$.
The event of changing from negative to positive is analogous.
We do the following reconfiguration, where we note that steps 6--8 are analogous to steps 2--4, but in reverse order:
\begin{enumerate}
\item
Reconfigure the top squares in positive OR gadgets where $x$ appears, so that there is room for a pyramid from the auxiliary variable component of $x$.
This is possible as there must be two true variables in any positive clause where $x$ appears before $x$ changes to negative, so there is at most one raised pyramid.

\item
Raise the pyramids from the positive auxiliary variable components to the positive OR gadgets.

\item
Reconfigure the positive auxiliary variable components to minus.

\item
Raise the pyramids from the primary to the positive auxiliary variable components.

\item
Reconfigure the primary variable component to minus.

\item
Remove the pyramids from the primary to the negative auxiliary variable components.

\item
Reconfigure the negative auxiliary variable components to minus.

\item
Remove the pyramids from the negative auxiliary variable components to the negative OR gadgets.
There is now at most one pyramid to the gadget. 
\end{enumerate}
To verify the correctness of the reconfiguration of the 
variable components consult~\Cref{fig:rec-variable}.
\end{proof}

\begin{figure}[htb]
    \centering
    \includegraphics[width=.9\textwidth,page=18]{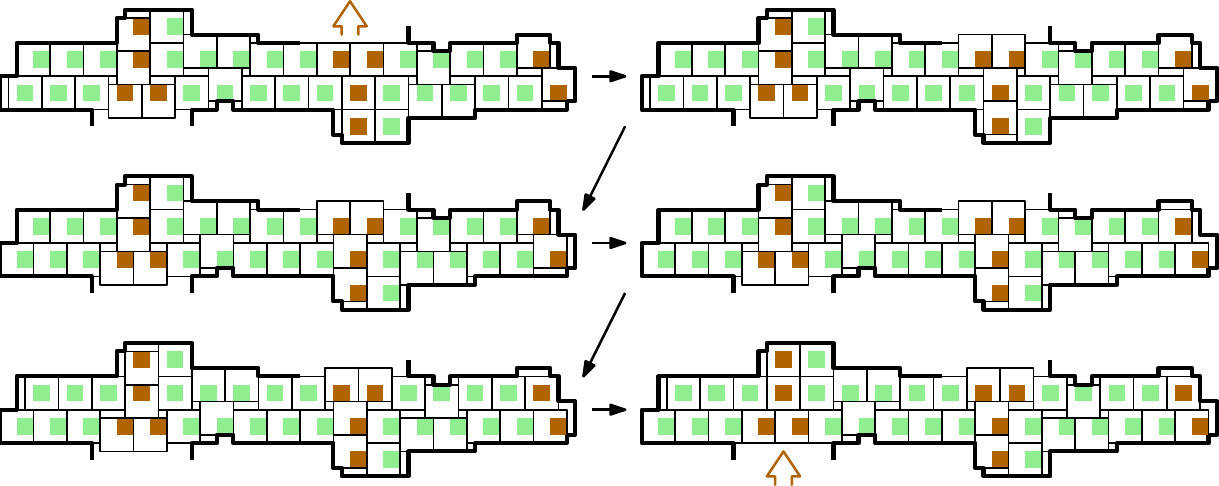}
    \caption{Reconfiguration of a variable component. Raising a pyramid pointing upwards allows to reconfigure the component and raising
    a pyramid pointing downwards.}
    \label{fig:rec-variable}
\end{figure}

\begin{lemma}\label{lem:dir2}
If there is a reconfiguration in $P$ from the start configuration $c_0$ to the target $c_1$ in $P$, then there is a reconfiguration in the formula $\phi$ from the satisfying assignment $\sigma_0$ to $\sigma_1$.
\end{lemma}
\begin{proof}
We show that we can assume that the squares move in time steps.
In each step, a single square moves from one integer position to a neighboring one by moving up, down, left or right.
In the end, we will see how the claim follows quite easily when we have a reconfiguration of that type.

Let $k$ be the number of squares, and consider an arbitrary reconfiguration from $c_0$ to $c_1$.
Let $x_i,y_i\colon [0,T]\rightarrow P$ be a parametrization of the center coordinates of square $i\in [k]=\{1,2,\ldots,k\}$, so that $(x_1,y_1,\ldots,x_k,y_k)$ describe the simultaneous motion of all squares from the start configuration at time $0$ to the target configuration at time $T$.
We now consider the movement we get by using $\lfloor x_i \rfloor$ instead of $x_i$ for all $i$.
This makes the $x$-coordinates constant at integers, except at discrete points in time where a subset of the squares $I\subseteq [k]$ ``jump'' from one integer $x$-coordinate to a neighboring one.
We then insert an extra time unit at these jump events, where the squares $I$ move horizontally at unit speed, keeping the $y$-coordinates fixed and keeping all the other squares $[k]\setminus I$ fixed.
Let $x'_i,y'_i\colon[0,T']\rightarrow P$ be the obtained reconfiguration, where each square again moves continuously from the start to the target configuration.

We claim that the obtained functions $x'_i,y'_i$ are also a valid reconfiguration, that is, we need to verify that the objects (i) stay within $P$ and (ii) don't collide with each other.
Part (i) follows because $P$ is a grid polygon.
To prove part (ii), suppose first that two squares $i,j\in [k]$ collide at a time that is not during a jump event, it is because there is a time $t'\in [0,T']$ where $\lvert x_i'(t)-x_j'(t)\rvert \in\{0,1,\ldots,7\}$ and $\lvert y_i'(t)-y_j'(t)\lvert < 8$.
Then there is also a corresponding time $t\in [0,T]$ where $\lvert x_i(t)-x_j(t)\rvert <8$ and $\lvert y_i(t)-y_j(t)\rvert < 8$.
Hence, there was also a collision in the original reconfiguration, which is a contradiction.
If the collision happens during a jump event, it can happen between two squares that move in the same direction (left or right) or in opposite directions or only one square moves.
In either case, there would also be a collision at the beginning or the end of the jump event, so the other argument applies again.

We can now apply the same modification to the $y$-coordinates, i.e., use $\lfloor y'_i\rfloor$ instead of $y'_i$ and insert time units at the events where these values jump to get continuous motion.
We obtain coordinate functions $x_i''$ and $y_i''$ where at every point in time, only a subset of the squares move, and they all move horizontally or vertically at unit speed from one integer position to a neighboring one.
Note that since the start and target configurations $c_0$ and $c_1$ use integer coordinates, our functions $x''_i,y''_i$ provide a reconfiguration between the same two configurations.

Consider a group of squares moving simultaneously, say horizontally.
There will in general be some moving right and the rest will move left.
We can move the ones moving right first and then move the ones moving left.
Thus, we obtain  a reconfiguration where all moving squares move simultaneously  in the same direction.
For a group of squares moving right, we can move them one by one in order of decreasing $x$-coordinates, and similarly for the other directions.
We have then obtained a reconfiguration with the claimed properties.

It remains to show how we can obtain a reconfiguration from $\sigma_0$ to $\sigma_1$ of $\phi$ from such a reconfiguration of the squares in $P$.
By \Cref{lem:implication,lem:or}, the PUSH gadgets and OR gadgets ensure that for each OR gadget, one of the participating variables does not make a pyramid.
Hence, for a positive clause, a primary variable must be plus; otherwise if it is minus or transitional, all variables would make pyramids.
Likewise, for a negative clause, a primary variable must be minus.
Therefore, each clause is satisfied by at least one of the variables, so any intermediate configuration from $c_0$ to $c_1$ corresponds to a satisfying assignment to $\phi$.
Since the squares move one at a time, only one variable changes at a time.
We conclude that we obtain a reconfiguration of $\phi$.
\end{proof}

Finally, we prove the containment in $\PSPACE$.

\begin{lemma}\label{lem:containment}
    Unlabeled reconfiguration of $8\times 8$ squares in simple grid polygons is contained in $\PSPACE$.
\end{lemma}
\begin{proof}
We observe first that the largest coordinate of the polygon
$P$ is bounded by $2^n$, where $n$ is the length of the
encoded reconfiguration instance.

We reuse the idea mentioned in the proof of~\Cref{lem:dir2}
that any reconfiguration can be turned into a reconfiguration
where only one square moves at a time and all static squares 
align with the grid. To store a configuration, where
every square aligns with the grid, we store the 
center coordinates of every square. 
Let $k< n$ be the number of squares. 
In this case we have no more than $({2^{2n}})^{k}< 4^{(n^2)}$ 
configurations that are aligned with the grid.
Furthermore, they can be efficiently represented.
Consider
a graph whose vertices are given by these configurations and
where two vertices are adjacent if the corresponding
configurations differ by a square that moved one unit.
The reconfiguration problem is then reduced to 
an undirected  $st$-reachability problem, which lies in
$\mathsf{LOGSPACE}$~\cite{Reingold}. Thus we need
at most polynomial space.
\end{proof}

Our main result now follows from \Cref{thm:monoton3SATpspace,lem:dir1,lem:dir2,lem:containment}.

\Squares*

As we only consider perfect configurations in our reduction, each square contains it reference center all the time by \cref{lem:refCenters}. Thus, we can easily label each square and obtain hardness for the labeled variant. 

\SquaresCor*

\section {Unit disks in polygons with holes}\label{sec:unitdisks}

 Recently, it was proven that the unlabeled motion planning problem for disk robots with two different radii in a domain with obstacles is $\PSPACE$-hard~\cite{DBLP:conf/fun/BrockenHKLS21}.
We turn our attention to unit disk robots and show that the problem remains $\PSPACE$-hard.

Before proving our result unit disks in a polygon with holes, we prove it for an \emph{arc-gon}, that is, a shape whose boundary consists of straight line segments and circular arcs.

\begin{proposition}\label{prop:arcgons}
Reconfiguration of unlabeled unit disk robots in an arc-gon is $\PSPACE$-hard. 
\end{proposition} 

\begin{proof}
 
We reduce from the 
configuration-to-configuration version of \textsc{Nondeterministic Constrained Logic} (NCL) \cite{DBLP:journals/tcs/HearnD05, DBLP:conf/icalp/HearnD02,HearnDemaineBook}. 
An NCL instance is given by a directed planar graph and two feasible orientations $G$ and $G'$. The edges are
colored blue or red, where every vertex is either incident to three blue edges
or to one blue and two red edges. A feasible orientation must respect 
an indegree of at least two at every vertex, where blue edges count with multiplicity 2. An NCL-instance is positive, if and only if 
there is a sequence of edge-reorientations that transforms $G$ into $G'$ and any intermediate orientation is
feasible.  

From a given NCL instance, we construct an arc-gon together with a one-to-one correspondence between configurations of unit disks within the arc gon, and configurations of the NCL instance, with the property that two configurations of the NCL instance are reachable from each other if and only of the corresponding configurations of unit disks can be reconfigured into one another.

\begin{figure}[htb]
    \centering
         \includegraphics[page=5, width=\textwidth]{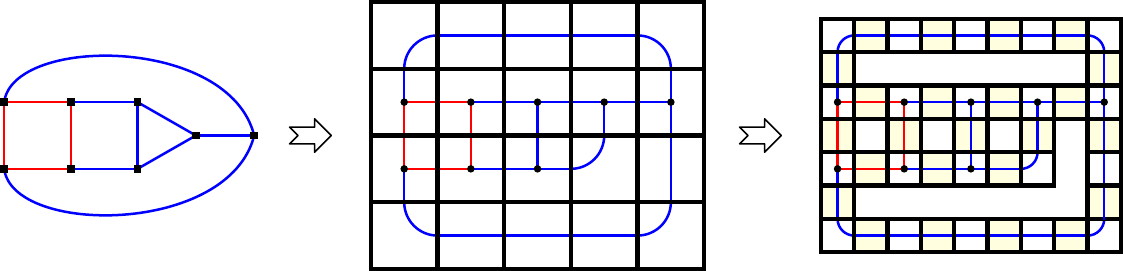}  
    \caption{A layout of an NCL instance using a limited set of tiles.}
    \label{fig:tiles}
\end{figure}

Following Brocken, van der Heijden, Kostitsyna, Lo-Wong and Surtel~\cite {DBLP:conf/fun/BrockenHKLS21}, we begin by laying out $G$ using a finite set of {\em tiles}; for instance, using the grid-embedding algorithm for cubic graphs by Tamassia~\cite {T87}, see \cref {fig:tiles} (left).
At this point, we include an additional step compared to~\cite {DBLP:conf/fun/BrockenHKLS21}: we shrink all tiles by a factor of $2$ and place additional {\em straight} tiles in the gaps. This increases the total number of tiles by at most a factor of $2$ but ensure that there is at least one straight tile between every pair of vertices of $G$, see \cref {fig:tiles} (right).

Specifically, there are only five different types of tiles that are used in such a layout (up to symmetry and grouping the blue and red connection tiles together), see \cref {fig:fourtypes}: we have {\em straight} tiles that contain a single piece of and edge of $G$ that enters and leaves through opposite sides of the tile; we have {\em turn} tiles that contain a single piece of an edge of $G$ that enters and leaves through adjacent sides of the tile; we have {\em AND} tiles that contain a single {\em AND} vertex of $G$ in two flavors depending on whether the red edges enter the tile through adjacent or opposite sides of the tile; and we have {\em OR} tiles that contain a single {\em OR} vertex of $G$. These tiles will be represented by four different geometric constructions. The geometric constructions are then concatenated and together will form a single polygonal domain.

\begin{figure}[htb]
    \centering
         \includegraphics[page=6]{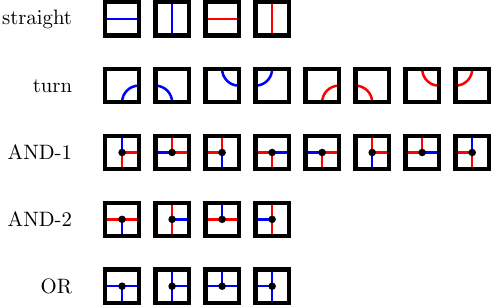}  
    \caption{Five types of tiles are needed.}
    \label{fig:fourtypes}
\end{figure}

In~\cite {DBLP:conf/fun/BrockenHKLS21}, each tile is realized as a {\em room}, with a designated disk functioning as a {\em door} between adjacent rooms. We follow this general idea: each tile is represented by a $6 \times 6$ interior space surrounded by a wall of thickness $1$. A door is represented by a $1 \times 2$ gap in the wall, in which a single unit disk fits. For an illustration, consider \cref {fig:room}.

\begin{figure}[htb]
    \centering
         \includegraphics[page=1,scale=.8]{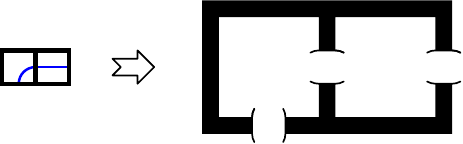}  
    \caption{Tiles become rooms, and rooms for adjacent tiles share an opening if the graph connects through the corresponding side of the tile (irrespective of its color). }
    \label{fig:room}
\end{figure}

 We further restrict the movement of this disk by adding four small ``hooks'' into the interiors of the two rooms in such a way that our door disk must stay in touch with the wall, and has freedom to move by exactly $1$ unit from one extreme position in which half of the disk is in the wall and half is interior to the room, to another extreme position in which the other half of the disk is in the wall and the other half is interior to the adjacent room, as depicted in \cref {fig:door}.

\begin{figure}[htbp]
    \centering
         \includegraphics[page=3,scale=.8]{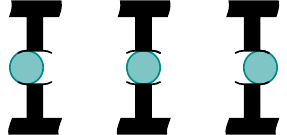}  
    \caption{All doors have a {\em door disk} that must stay within the door opening. It has a continuous segment of valid placements of length $1$. Door disks are marked in {\em green} in all figures.}
    \label{fig:door}
\end{figure}

For the interiors of the rooms, we present carefully constructed designs that mimic the behavior of the vertices and edges of $G$ during a reconfiguration.
Our construction closely follows~\cite {DBLP:conf/fun/BrockenHKLS21}.
One main difference is that we require a gadget that ensures an edge can only have one orientation at a time, which we implement in our {\em straight} tile gadget.
Our designs are listed in \cref {fig:tilegeometry}.

We call the two extreme states of a door disk \emph{exterior} and \emph{interior}; all other positions are \emph{intermediate}.
A priori, multiple door disks of a room could be in an intermediate state. To avoid this, we make sure that each pair of tiles is separated by a {\em straight} tile. We now argue that at least one door disks of a straight tile is exterior. 

\begin {lemma} \label {lem:extreme}
 In every valid configuration of the room of a  straight tile, at least one door disk is exterior. Moreover, in any reconfiguration of a door disk from the  exterior state to the interior state, there is a point in time when both door disks are exterior. 
\end {lemma}

\begin{figure}[htbp]
    \centering
         \includegraphics[page=2,scale=.8]{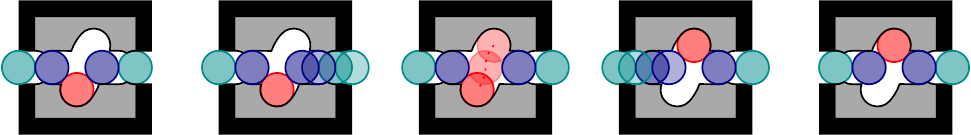}  
    \caption{A sequence of moves that imitates the reorientation of an edge. }
    \label{fig:reorient}
\end{figure}

\begin {proof}
A move sequence corresponding to reorienting the edge of a {\em straight} tile is illustrated in \cref {fig:reorient}.
While the center of the red disk is in the bottom half of the room, the left door disk must remain in its exterior state. Symmetrically, while the center of the red disk is in the top half of the room, the right door disk must remain in its exterior state. Therefore, the left and right door disk cannot both be at an intermediate state simultaneously.
\end {proof}

Using Lemma~\ref {lem:extreme}, we may now assume that all door disks are in extremal positions while we are in the process of reorienting edges; hence, we can consider each reorientation individually.
This allows us to fully categorize the feasible configurations inside each tile, see also \cref {fig:tilegeometry}:
\begin {itemize}
  \item {\bf Straight} tile: Exactly the configurations with both door disks exterior, or with one exterior and one interior, are feasible.
  \item {\bf Turn} tile: Exactly the configurations with both door disks exterior, or with one exterior and one interior, are feasible.
  \item {\bf AND-1} tile: If the top door disk is {\em interior}, then the only feasible configuration has both other door disks exterior. If the top door disk is {\em exterior}, then all configurations of the other door disks are feasible.
  \item {\bf AND-2} tile: If the top door disk is {\em interior}, then the only feasible configuration has both other door disks exterior. If the top door disk is {\em exterior}, then all configurations of the other door disks are feasible.
  \item {\bf OR} tile: Exactly those configurations with at least one exterior door disk are feasible.
\end {itemize}

\begin{figure}[htb]
    \centering
         \includegraphics[page=5,scale=.8]{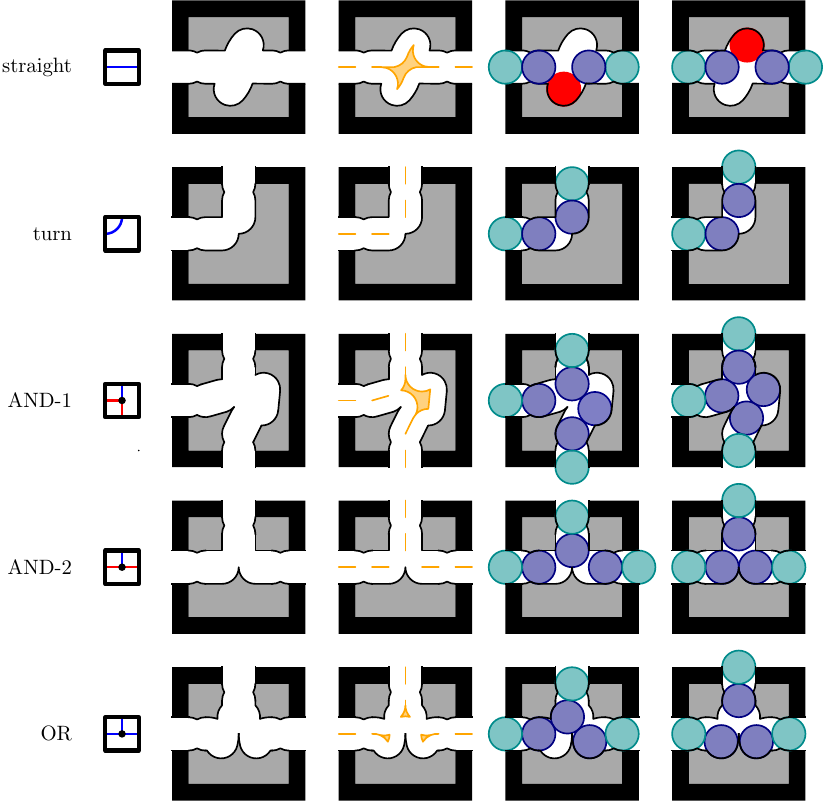}  
    \caption{Geometric realizations of the five types of tiles are needed. Marked in orange is the {\em dilation} of the shapes by $1$; these are exactly the locations where disk centers can be without intersecting the boundary. On the right two different valid configurations of disks with different extremal positions of the door disks.}

    \label{fig:tilegeometry}
\end{figure}

For the AND and OR tiles, which represent vertices of $G$, this behavior exactly models the constraint that the in-degree should be at least $2$, when red edges count as $1$ and blue edges count as $2$, and when an arrow pointing {\em into} a vertex corresponds to the door disk being {\em exterior}.
The straight and turn tiles represent (pieces of) edges of $G$, and their constructions ensure that a single edge can never be oriented in both directions simultaneously; if a door disk is exterior for an AND or OR tile, it is interior for the incident straight tile. Note that they do allow for a configuration where the edge is not pointing in either direction; indeed, such configurations are necessary as intermediate configuration while we reorient an edge.
\end{proof}

Then, we turn the arc-gon constructed in the proof of \Cref{prop:arcgons} into a polygon with holes and thereby complete the proof of \cref{thm:disks}. Observing that no two disks may swap within the polygon implies \PSPACE-hardness for the labeled variant, i.e.,  \cref{cor:disksLabeled}.

To complete the proof of  complete the proof of \cref{thm:disks}, it remains to show that we can turn the arc-gon of \Cref{prop:arcgons} into a polygon with holes.

\DisksUnlabeled*

\begin{proof}

The construction in \Cref{prop:arcgons} results in an arc-gon with the following properties: In addition to line segments (infinite radius arcs) there are arcs of three different radii in the construction:
\begin {itemize}
  \item Arcs of radius $1$, which create pockets for disks to sit in. These occur in all tile designs; for example, the little hooks at the doors (Figure~\ref {fig:door}) are arcs of radius $1$. Arcs of radius $1$ in the boundary correspond to single points in the placement space for the disks.
  \item Arcs of radius $2$, which constrain disks to rotate around a point on their boundary. For instance, the OR tile utilizes arcs of radius $2$ in the bottom part. Arcs of radius $2$ in the boundary correspond to arcs of radius $1$ in the placement space for the disks.
  \item Arcs of radius $3$, which constrain disks to roll around other (stationary) disks. Arcs of radius $3$ are employed in e.g. the {\em straight} tile design: when the blue disks are in exterior positions, the red disk can switch state by rolling first around one and then the other disk. Arcs of radius $3$ in the boundary correspond to arcs of radius $2$ in the placement space for the disks.
\end {itemize}
We now adapt the construction to not require any arcs on the boundary.

First, we claim that we can safely replace all arcs of radius $1$ by polylines that approximate the arcs from the outside, without changing the behavior of the construction at all. Indeed, a disk that is in contact with an arc of radius $1$ cannot move in any direction represented by the arc, and this remains the case if we replace the arc by a discrete set of contact points (so long as they are spaced by less than $\pi$).

Next, we claim that we can also replace the arcs of radius $2$ or $3$ by polylines; however, this will increase the feasible space for the disks, and thus we need to argue more carefully.

Specifically, we replace all arcs of radius $3$ by polylines that approximate these arcs from the outside, by placing $k$ contact points evenly spaced around the arcs. This now allows the disks to move by less than a distance of $\delta$ into the newly created pockets, thus releasing contact with the blue disk, which in turn can move by the same amount, finally allowing the door disk to also move by the same amount.

We set $k = 25$, which implies that $\cos (\frac {2\pi}{25}) \approx \frac {3}{3.1}$ so $\delta \approx 0.1$, which is sufficient.
Now, this implies we have only a weaker version of Lemma~\ref {lem:extreme}: we can no longer ensure that door disks are in extreme positions, but rather than they are within distance $0.1$ from an extreme position. This means we have to revisit the possible configurations within our tiles, but they still only allow the states as listed before.

Finally, we also replace the circular arcs of radius $2$ by polylines; again choosing $k=25$ will create a slack of at most $0.1$ (in fact fewer points are needed since the arcs are smaller).
\end{proof}

In order to obtain \cref{cor:disksLabeled} it suffices to observe  the polygon is thin enough to guarantee that no two disks may swap. This implies the following result.

\DisksLabeled*

\section{Discussion}\label{sec:discussion}
We proved that reconfiguration of \MPTSAT is $\PSPACE$-hard, and thereby complemented a sequence of recent results establishing the relationships between \TSAT variants and their reconfiguration counterparts.
This result was a stepping stone towards showing $\PSPACE$-hardness of reconfiguration of unit squares in a simple grid polygon.
While the reduction builds on ideas from Abrahamsen and Stade~\cite{DBLP:journals/corr/abs-2404-09835},  we also needed several new ideas. 
We  highlight some notable differences:
Many squares in the packings from~\cite{DBLP:journals/corr/abs-2404-09835} are \emph{locked}, so there is no possible reconfiguration between two packings that correspond to two satisfying assignments where just a single variable differs in truth value.
For instance, it is not possible to move the squares in the variable component so that they change between the states \emph{plus} and \emph{minus}, which are used to represent the truth values of a variable.
Our reduction uses a modified construction for variables with plus and minus states and an additional \emph{transitional} state that can get from one to the other.

A similar issue arises with respect to the pyramids, which in~\cite{DBLP:journals/corr/abs-2404-09835} are always present from the primary to the auxiliary variable components and from the auxiliary variable components to the OR gadgets.
An extra column of \emph{push squares} is raised towards the positive clauses in case a variable was minus.
However, it is not possible to reconfigure between the packings where the push squares are raised and where they are not.
Instead, in this paper, we allow the entire pyramid to be created or removed.
A pyramid is created when all the squares of the pyramid raise or fall compared to the surrounding squares.
The functionality that pyramids are only occasionally present requires the use of the \emph{helper rows} in the variable component, which therefore more involved than in~\cite{DBLP:journals/corr/abs-2404-09835}.
Finally, the OR gadget also required adjustments, so that we can reconfigure from one satisfying assignment to another.

The construction from~\cite{DBLP:journals/corr/abs-2404-09835} resulted in $\NP$-hardness of packing $2\times 2$ squares in grid polygons, where a \emph{grid polygon} is an orthogonal polygon with corners at integer coordinates.
The improved flexibility required in the present paper, for the squares to be able to reconfigure, comes at the cost of a finer granularity of the polygon compared to the size of the squares:
We show $\PSPACE$-hardness of reconfiguration of $8\times 8$ squares in a grid polygon. 

We have also shown that reconfiguration of unit disks in a polygon with holes is $\PSPACE$-hard.
A natural next question is unit disks in a simple polygon. At a first glance, it seems difficult to extend our insights in the square setting to this case, because it is troublesome to propagate movement independently in multiple directions in a large region filled with disks. Therefore, we expect that a $\PSPACE$-hardness proof, if it exists, will likely require different techniques.

\bibliography{lib.bib}
\end{document}